\documentclass[12pt,twoside]{amsart}
\usepackage{mathrsfs,amsthm,amscd,amsmath,amssymb}
\usepackage[colorlinks,linkcolor=green,citecolor=blue, pdfstartview=FitH]{hyperref}
\usepackage{amscd}
\usepackage{actuarialsymbol}
\usepackage{amsfonts}
\usepackage{graphicx}
\usepackage{caption, subcaption}

\title
[Distortion risk measures of sums of two counter-monotonic risks]
{Distortion risk measures of sums of two counter-monotonic risks}
\author{Chunle Huang}
\address{Chunle Huang, Institute of Mathematics, Hunan University, Changsha, 410082, China.}
\email{chunle@zju.edu.cn; 402961544@qq.com}
\newtheorem{thm}{Theorem}[section]
\newtheorem{lem}[thm]{Lemma}
\newtheorem{cor}[thm]{Corollary}

\theoremstyle{definition}
\newtheorem{defn}[thm]{Definition}

\newtheorem{ex}[thm]{Example}
\makeatletter
\let\uppercasenonmath\@gobble% disables title uppercase
\makeatother
\begin{document}
\bibliographystyle{amsalpha+}
\maketitle
\begin{abstract} 
In this paper, we will show that under certain conditions, associated to any fixed distortion function $g$, the distortion risk measure of a sum of two counter-monotonic risks can be expressed as the sum of two related distortion risk measures of the marginals involved, one associated to the original distortion function $g$ and the other  associated to the dual distortion function of $g$. This result extends some of the work in \cite{Chaoubi et al. (2020)} and \cite{HLD} since the class of distortion risk measures includes the risk measure of VaR and TVaR as special cases. 
\end{abstract}

%\tableofcontents 
\section{Introduction} 
A risk measure is by definition a mapping from the set of random variables, usually representing the risks at hand, to the set of real numbers $\mathbb{R}$. Risk measures are a helpful tool for decision making because they are able to summarize the information available about the random variable $X$ into one single number $\rho[X]$. Common risk measures in actuarial science are premium principles; see for instance Chapter 5 in \cite{KGDD2008}. Other risk measures, which measure the upper tails of distribution functions, are used for determining provisions and capital requirements of an insurer, in order to avoid insolvency. Such measures of risk are considered in \cite{ADE1999}, \cite{Dhaene et al. (2006)}, \cite{WH2000}, and \cite{W1996} among others. In this paper, we will mainly concentrate on risk measures that can be used for reserving and solvency purposes. 

Along this line, one fundamental risk measure is VaR, which is defined as $\text{VaR}_p[X] = F^{-1}_X(p)$ for any risk $X$ and a probability level $p \in (0, 1)$. VaR has a simple interpretation in terms of over- or undershoot probabilities and in today’s financial world, it has become the benchmark risk measure. However, a single VaR at a predetermined level $p$ does not give any information about the thickness of the upper tail of the distribution function under consideration. This is a considerable shortcoming since in practice a regulator is often not only concerned with the frequency of default, but also with the severity of default. Also shareholders and management should be concerned with the question ‘how bad is bad?’ when they want to evaluate the risks at hand in a consistent way. To overcome this problem, TVaR can be used which is defined as
$\text{TVaR}_p[X] = \frac{1}{1 - p}\int^1_{p}F^{-1}_X(q)dq$ for any risk $X$ and a probability level $p \in (0, 1)$. 
We remark that both VaR and TVaR are important in practice since regulators accept these measures as the basis for setting capital requirements for market risk exposure. 

However, from the definition of TVaR, the TVaR risk measure only uses the upper tail of the distribution. Hence, this risk measure does not create incentive for taking actions that increase the distribution function for outcomes smaller than $Q_p$. Also, we remark that TVaR only accounts for the expected shortfall and hence, does not properly adjust for extreme low-frequency but high severity losses. Wang in \cite{W1996} and \cite{Wang2000} independently introduced a family of risk measures as a solution to these problems, which is now called the class of distortion risk measures,  by using the concept of distortion function as introduced in Yaari's dual theory of choice under risk; see also \cite{W2001} and \cite{WY1998}. In the literature, the distortion risk measure associated with a distortion function $g$ is usually denoted by $\rho_g[X]$ for any random variable $X$. Note that the distortion function $g$ is assumed to be independent of the distribution function of the random variable $X$. More and more attention has been paid to the class of distortion risk measures since Wang's work; see for instance \cite{DDV1999}, \cite{Dhaene et al. (2012)}, \cite{Dhaene et al. (2006)} and \cite{H2004}. One reason for its growing popularity is that by choosing specific distortion functions one can easily obtain VaR and TVaR, in other words, the class of distortion risk measures includes both VaR and TVaR as special cases. Another reason is the correspondence between the two common theories of choice under risk, namely the expected utility theory and Yaari's dual theory of choice under risk. For instance, one can express risk-averse decision makers using the concept of distortion function which, in turn, defines a particular type of distortion risk measures. The third reason is that one can express some of the common stochastic orderings among risks, stop-loss and convex orderings in particular, in terms of distortion risk measures.  

One concept that can be used to study different risk measures is comonotonicity, which refers that the components of a random vector $\underline{X} = (X_1, X_2, ..., X_n)$ can all be expressed as non-decreasing functions of the same random variable, so, it corresponds with the riskiest dependence structure in a given Fréchet space $\mathcal{R}_n(F_1, F_2, ..., F_n)$. 
%Over the last decades, researchers in economics, financial mathematics and actuarial science have obtained a lot of results related to the concept of comonotonicity in their respective fields of interest.  
Comonotonicity has been extensively discussed in \cite{Dhaene et al. (2002a)} and \cite{Dhaene et al. (2002b)} and for its numerous applications in finance and actuarial science we recommend the reader to refer to \cite{DDV2011}, \cite{DD2007}, \cite{DVG2005}, \cite{DWY2000}, \cite{WD1998}, \cite{WYP1997}. For our purposes, we point out that an important result regarding comonotonicty is that distortion risk measures of the comonotonic sum ${S}^c$ can be expressed as the sum of the corresponding distortion risk measures of the marginals involved, that is, the formula $\rho_g[S^c] = \sum^n_{i = 1}\rho_g[X_i]$ holds for any distortion function $g$, see \cite{Dhaene et al. (2012)}, \cite{Dhaene et al. (2006)} and \cite{W1996} for more details. In this paper, we want to investigate the additivity property for sums of two counter-monotonic risks. 

As a natural generalization of comonotonicty, the concept of counter-monotonicity appears as a bivariate negative dependence between variables. However, much less attention in the literature  has been devoted to this concept, though the authors in \cite{DD2003} obtained very early the simple characterizations of comonotonicity and counter-monotonicity by extremal correlations. One possible reason is that it is not easy to express the VaR of a counter-monotonic sum in terms of the VaR’s of the marginal components of the sum. Indeed, in a financial context, counter-monotonicity may be useful to assess the merge of dependent assets. For example, Cheung et al. \cite{CDL2014} investigated under what conditions, the sum of counter-monotonic risks will decrease the overall risk. Recently, Chaoubi et al. \cite{Chaoubi et al. (2020)} obtained closed-form expressions for the VaR and TVaR of the sum of two counter-monotonic risks for several specific choices of marginal distributions, including symmetric or unimodal distributions. They showed that under certain conditions, the VaR and TVaR of the counter-monotonic sum can be expressed as the sum of several related individual measures. Hanbali et al. \cite{HLD} generalized the main results of Chaoubi et al. \cite{Chaoubi et al. (2020)} by providing explicit decompositions of the VaR, TVaR and stop-loss premium of
counter-monotonic sums with general marginal distributions.

Following the work of several authors in the literature, we will derive in this paper closed-form decompositions of the distortion risk measures of counter-monotonic sums with suitable marginal distributions. In particular, we will show that under certain conditions, associated to any fixed distortion function $g$, the distortion risk measure $\rho_g[S^{-}]$ of the counter-monotonic sum $S^{-}$ can be expressed as the sum of two related distortion risk measures of the marginals involved, one associated to the original distortion function $g$ and the other  associated to the dual distortion function of $g$. This result extends some of the work of Chaoubi et al. \cite{Chaoubi et al. (2020)} and Hanbali et al. \cite{HLD} from the case of VaR and TVaR to the general case of distortion risk measures and contributes to the body of literature devoted to the decomposition of risk measures of sums of two random variables; see for instance \cite{CDD2017}, \cite{CL2013}, \cite{DD1999}, \cite{Dhaene et al. (2006)} and \cite{Dhaene et al. (2012)} among others.

The remainder of the paper is organized as follows. In Section 2, we introduce the notations and relevant definitions. In Section 3, we prove the main result. Section 4 presents several applications of the main result. Section 5 concludes the paper. 

\section{Preliminaries}
\subsection{Quantile functions} 
All random variables are defined on a common probability space $(\Omega, \mathcal{F}, \mathbb{P})$ and it is always assumed that all random variables are such that the risk measures introduced hereafter are finite. 

The cumulative distribution function (cdf) of a random variable $X$ is denoted by $F_X$. The left inverse of $F_X$ is denoted by $F^{-1}_X$, and is defined as $$F^{-1}_X(p) = \inf\{x \in \mathbb{R} \,\,|\,\, F_X(x) \geq p\}, \,\, p \in [0, 1]$$ with the convention $\inf \emptyset = +\infty$. The right inverse of $F_X$ is denoted by $F^{-1+}_X$, and is defined as $$F^{-1+}_X(p) = \sup\{x \in \mathbb{R} \,\, | \,\, F_X(x) \leq p\}, \,\, p \in [0, 1]$$ with the convention $\sup\emptyset = -\infty$. If $F_X$ is continuous and strictly increasing, the left and the right inverses are equal. Otherwise, on horizontal segments of $F_X$, the inverses $F^{-1}_X(p)$ and $F^{-1+}_X(p)$ are different and they satisfy $F^{-1}_X(p) < F^{-1+}_X(p)$. In such cases, the generalized $\alpha$-inverse $F^{-1(\alpha)}_X$, for any $\alpha \in [0, 1]$, is useful, which is defined as: 
$$F^{-1(\alpha)}_X(p) = (1 - \alpha)F^{-1}_X(p) + \alpha F^{-1+}_X(p), \,\, p \in [0, 1].$$ For more details on generalized inverses we recommend the interested reader to refer to \cite{DDGK}, \cite{Dhaene et al. (2002a)}, \cite{EH2013} and \cite{KRK2015} among others. 
\subsection{Distortion risk measures} 
In this subsection, we recall the class of distortion risk measures, introduced by Wang \cite{W1996}. The quantile risk measure and TVaR belong to this class. The expectation of $X$ if it exists, can be written as 
$$E[X] = -\int^0_{-\infty}[1 - \overline{F}_X(x)]dx + \int^{\infty}_0 \overline{F}_X(x)dx.$$ Wang \cite{W1996} defines a family of risk measures by using the concept of distortion function as introduced in Yaari's dual theory of choice under risk; see also \cite{WY1998}. A distortion function is defined as a non-decreasing function $g: [0, 1] \to [0, 1]$ such that $g(0) = 0$ and $g(1) = 1$. The distortion risk measure associated with distortion function $g$ is denoted by $\rho_g[X]$ and is defined by 
$$\rho_g[X] = -\int^0_{-\infty}[1 - g(\overline{F}_X(x))]dx + \int^{\infty}_0 g(\overline{F}_X(s))dx$$ 
for any random variable $X$. Note that the distortion function $g$ is assumed to be independent of the distribution function of the random variable $X$. The distortion function $g(q) = q$ corresponds to $E[X]$. Also note that $g_1(q) \leq g_2(q)$ for all $q \in [0, 1]$ implies that $\rho_{g_1}[X] \leq \rho_{g_2}[X]$. 

We remark that Dhaene et al in \cite{Dhaene et al. (2012)} and \cite{Dhaene et al. (2006)} recently obtained Lebesgue-Stieltjes integral representations of distortion risk measures, see Theorem 4 and Theorem 6 in \cite{Dhaene et al. (2012)}, which is of interest and very useful because it enables us to treat distortion risk measures in a very simple manner. Using this representation they fully showed for general case the additivity of distortion risk measures for comonotonic risks. This property is well-known but, unfortunately, the proofs presented in the literature are often incomplete, in the sense that they only hold for a particular type of distortion functions, such as the class of concave distortion functions. We provide a number of distortion risk measures as follows. The reader is strongly referred to \cite{Dhaene et al. (2012)} and \cite{Dhaene et al. (2006)}. 

\begin{ex}[VaR is a distortion risk measure]\label{0080}
Recall that the Value-at-Risk (VaR) of a random variable $X$ at level $p \in (0, 1)$ is defined by $\text{VaR}_p[X] = F^{-1}_X(p)$. For any fixed $p \in (0, 1)$, we consider the function $g$ defined by 
$$g(q) = \mathbb{I}(q > 1 - p), \,\, q \in [0, 1].$$ 
Then it is easy to check that $g$ is a left-continuous distortion function. Therefore, by Theorem 6 in \cite{Dhaene et al. (2012)}, for any random variable $X$, the distortion risk measure $\rho_{g}[X]$ has the following Lebesgue-Stieltjes integral representation: $$\rho_{g}[X] = \int_{[0, 1]}F^{-1}_{X}(1 - q) dg(q),$$ which implies that 
$$ 
\rho_g[X] = \lim_{q \to (1 - p)^{+}}F^{-1}_{X}(1 - q)g(q) - \lim_{q \to (1 - p)^{-}}F^{-1}_{X}(1 - q)g(q) = \text{VaR}_p[X].
$$ 
\end{ex}
\begin{ex} [TVaR is a distortion risk measure] \label{0088} 
Recall that the Tail-Value-at-Risk (TVaR) of a random variable $X$ at level $p \in (0, 1)$ is defined by $\text{TVaR}_p[X] = \frac{1}{1 - p}\int^1_{p}F^{-1}_X(q)dq$. For any fixed $p \in (0, 1)$, we consider the function $g$ defined by $$g(q) = \min\Big\{\frac{q}{1 - p}, 1\Big\}, \,\, q \in [0, 1].$$ Then it is easy to check that $g$ is a continuous distortion function. Therefore, by Theorem 6 in \cite{Dhaene et al. (2012)}, we have for any random variable $X$ $$\rho_g[X] = \int_0^1\text{VaR}_{1 - q}[X]dg(q) 
= \frac{1}{1 - p}\int_0^{1 - p}\text{VaR}_{1 - q}[X]dq 
= \text{TVaR}_{p}[X].$$
\end{ex}
\begin{ex} [The Wang Transform Risk Measure] 
The Wang Transform risk measure was introduced by Wang \cite{Wang2000} as an example of distortion risk measures. For any fixed $p \in (0, 1)$, we consider the distortion function $g$ defined by $$g(q) = \Phi[\Phi^{-1}(q) + \Phi^{-1}(p)], \quad q \in [0, 1].$$ Then it is to check that $g$ is a continuous distortion function, which is usually called the Wang Transform (distortion function) at level $p$. The corresponding distortion risk measure $\rho_g[X]$ is called the Wang Transform risk measure, denoted by $\text{WT}_p[X]$. More details about this kind of distortion risk measure can be found in \cite{Wang2000}. 
\end{ex}

\subsection{ Fréchet spaces and Fréchet bounds} 
Let $F_1, F_2, ..., F_n$ be univariate cumulative distribution functions (c.d.f.'s, in short) and consider the Fréchet space $\mathcal{R}_n(F_1, F_2, ..., F_n)$ consisting of all $n$-dimensional c.d.f.'s $F_X$, or equivalently of all the $n$-dimensional random vectors $X = (X_1, X_2, ..., X_n)$ possessing $F_1, F_2, ..., F_n$ as their marginal c.d.f.'s. For all $X$ in $\mathcal{R}_n(F_1, F_2, ..., F_n)$, we have the following inequality 
$$M_n(x) \leq F_X(x) \leq W_n(x)$$
for all $x = (x_1, x_2, ..., x_n) \in \mathbb{R}^n$, where $W_n$ is usually referred to as the Fréchet upper bound of $\mathcal{R}_n(F_1, F_2, ..., F_n)$ and is defined by 
$$W_n(x) = \min\{F_1(x_1), F_2(x_2), ..., F_n(x_n)\}, \,\, x \in \mathbb{R}^n,$$
while $M_n$ is usually referred to as the Fréchet lower bound of $\mathcal{R}_n(F_1, F_2, ..., F_n)$ and is defined by 
$$M_n(x) = \max\Big\{\sum^n_{i = 1}F_i(x_i) - n + 1, 0\Big\}, \,\, x \in \mathbb{R}^n.$$

$W_n$ is reachable in $\mathcal{R}_n(F_1, F_2, ..., F_n)$. Indeed, for a random variable $U$, uniformly distributed on $[0, 1]$, it can be shown that $W_n$ is the c.d.f. of the random vector 
$$(F^{-1}_1(U), F^{-1}_2(U), ..., F^{-1}_n(U)) \in \mathcal{R}_n(F_1, F_2, ..., F_n).$$
The elements of the Fréchet space $\mathcal{R}_n(F_1, F_2, ..., F_n)$ which have a multivariate c.d.f. given by $W_n(x)$, are said to be comonotonic, see the following subsection. On the contrary, when $n \geq 3$, $M_n$ is not always a c.d.f. anymore, as shown by Tchen \cite{Tchen1980}. Theorem 3.7 in \cite{Joe1997} provides a necessary and sufficient condition for $M_n$ to be a c.d.f. in $\mathcal{R}_n(F_1, F_2, ..., F_n)$, though this condition is generally not easy to check. Inspired by a result in \cite{HW1999}, Dhaene and Denuit in \cite{DD1999} provided a more practical condition for $M_n$ to be a proper c.d.f. in $\mathcal{R}_n(F_1, F_2, ..., F_n)$ by introducing the concept of mutually exclusivity of random vectors. 

\subsection{Comomontonic and counter-monotonic risks} In this subsection, we briefly recall the concept of comonotonicity and counter-monotonicity, both of which consider extreme dependence relationships between variables. 
\begin{defn}
A random vector $\underline{X} = (X_1, X_2, ..., X_n)$ is said to be comonotonic, denoted by $\underline{X}^c = (X_1^c, X_2^c, ..., X_n^c)$, if $$\underline{X} \stackrel{d}{=} (F^{-1}_{X_1}(U), F^{-1}_{X_2}(U), ..., F^{-1}_{X_n}(U))$$ 
where $ \stackrel{d}{=}$ stands for equality in distribution and $U$ is a uniformly distributed random variable on $[0, 1]$. 
\end{defn}
By definition we know that the components of a comonotonic random vector $\underline{X}^c$ are jointly driven by a single random variable transformed by the non-decreasing functions $F^{-1}_{X_1}, F^{-1}_{X_2}, ..., F^{-1}_{X_n}$. The marginals of $\underline{X}^c$ are, however, equal in distribution to those of the original vector $\underline{X}$. Comonotonicity has been extensively discussed in \cite{Dhaene et al. (2002a)} and \cite{Dhaene et al. (2002b)}. For its applications in finance and actuarial science we recommend the reader to refer to \cite{DDV2011}, \cite{DD2007}, \cite{DVG2005}, \cite{DWY2000}, \cite{WD1998} and \cite{WYP1997}. An important result regarding comonotonicty is that distortion risk measures of the comonotonic sum $\underline{X}^c$ can be expressed as the sum of the corresponding distortion risk measures of the marginals involved, that is, the formula $\rho_g[S^c] = \sum^n_{i = 1}\rho_g[X_i]$ holds for any distortion function $g$, see \cite{Dhaene et al. (2012)}, \cite{Dhaene et al. (2006)} and \cite{W1996} for more details. 

\begin{defn} 
A random vector $\underline{X} = (X_1, X_2)$ of dimension two is said to be counter-monotonic, denoted by $\underline{X}^{-} = (X_1^{-}, X_2^{-})$, if there is a random variable $U$, uniformly distributed on $[0, 1]$, such that $$\underline{X} \stackrel{d}{=} (F^{-1}_{X_1}(U), F^{-1}_{X_2}(1 - U)).$$ 

It is easy to see that the counter-monotonic modification $\underline{X}^{-}$ of a random vector $\underline{X}$ has the same marginal distributions as $\underline{X}$, but the components of $\underline{X}^{-}$ move perfectly in the opposite direction. Hence, counter-monotonicity corresponds to extreme negative dependence structure. It is well known that counter-monotonicity leads to a convex lower bound for a sum of dependent risks, see for instance \cite{DD1999} and \cite{CL2014}. Let $S^{-}$ be the sum of the components of a counter-monotonic vector $\underline{X}^{-} = (X_1^{-}, X_2^{-})$, then we have $S^{-} \stackrel{d}{=} F^{-1}_{X_1}(U) + F^{-1}_{X_2}(1 - U)$ for a random variable $U$ which is uniformly distributed on $[0, 1]$. For more details about $S^{-}$ and its risk measures, we recommend the reader to see \cite{Chaoubi et al. (2020)} and \cite{HLD}. 
\end{defn}
\subsection{Dispersive order} In this subsection we recall the concept of dispersive order, which is useful when comparing univariate random variables by their variability. More details about this concept can be found for instance in \cite{DDGK}, \cite{Jeon et al. (2006)} and \cite{S-S}. 
\begin{defn} 
Let $X$ and $Y$ be two random variables with cumulative distribution functions $F_X$ and $F_Y$, respectively. If $F_X^{-1}(u) - F_X^{-1}(v) \leq F_Y^{-1}(u) - F_Y^{-1}(v)$, for any $0 < v \leq u < 1$, then $X$ is said to be smaller than $Y$ according to the dispersive order, denoted by $X \leq_{disp} Y$.  
\end{defn}
It is easy to check that $X \leq_{disp} Y \iff F_Y^{-1}(u) - F_X^{-1}(u)$ is increasing in $u \in (0, 1)$ and $X \leq_{disp} Y \iff X + c \leq_{disp} Y$ for any $c \in  \mathbb{R}$.  
\begin{ex}  
Assume that $(X_1, X_2)$ is a pair of normal random variables such that $X_i \sim N(\mu_i, \sigma^2_i), i = 1, 2$. Then $X_1 \leq_{disp} X_2$ if $\sigma_1 < \sigma_2$ and $X_2 \leq_{disp} X_1$ if $\sigma_1 \geq \sigma_2$. In fact, for any $p \in (0, 1)$, by Theorem 1 in \cite{Dhaene et al. (2002a)} or Property 1.5.16 in \cite{DDGK}, we have $$F_{X_1}^{-1}(p) = \mu_1 + \sigma_1\Phi^{-1}(p) \,\, \text{and} \,\, F_{X_2}^{-1}(p) = \mu_2 + \sigma_2 \Phi^{-1}(p).$$
It follows that $$F_{X_2}^{-1}(p) - F_{X_1}^{-1}(p) = (\mu_2 - \mu_1) + (\sigma_2 - \sigma_1)\Phi^{-1}(p),\,\, p \in(0, 1).$$
This implies that $X_1 \leq_{disp} X_2$ if $\sigma_1 < \sigma_2$. The proof of the case $\sigma_1 \geq \sigma_2$ is similar.
\end{ex}

\begin{ex} [See Example 1 in \cite{Chaoubi et al. (2020)} and Example 3 in \cite{Jeon et al. (2006)}]We provide below a list of other examples of pairs of random variables $(X_1, X_2)$ such that $X_2\leq_{disp} X_1$. 
\begin{enumerate} 
\item If $X_i \sim \text{Student}(\nu_i), i = 1, 2$, with degrees of freedom $0 < \nu_1 < \nu_2$, then $X_2 \leq_{disp} X_1$;  
\item If $X_1 \sim \text{Student}(\nu)$ with $\nu > 0$ and $X_2 \sim N(0, 1)$, then $X_2 \leq_{disp} X_1$;  
\item If $X_1 \sim \text{Laplace}(0, 2^{0.5})$ and $X_2 \sim N(0, 1)$, then $X_2 \leq_{disp} X_1$; 
\item If $X_1 \sim \text{Logistic}(0, 1)$ and $X_2 \sim N(0, 1)$, then $X_2 \leq_{disp} X_1$. 
\end{enumerate}
\end{ex}

\section{Main result}
In this section, we study distortion risk measures of sums of two counter-monotonic risks. The main result is the following theorem, which tells us that under some suitable conditions, associated to any distortion function $g$, the distortion risk measure $\rho_g[S^{-}]$ of the counter-monotonic sum $S^{-} = X^{-}_1 + X^{-}_2$ can be expressed as the sum of two individual distortion risk measures, one associated to the original distortion function $g$ while the other  associated to the dual distortion function of $g$.  

\begin{thm} \label{0265}
Assume that $(X^{-}_1, X^{-}_2)$ is a pair of symmetric random variables such that their distribution functions $F_{X_i}(x), i = 1, 2$, are continuous and strictly increasing on $(F^{-1+}_{X_{i}}(0), F^{-1}_{X_i}(1))$. If $X_2 \leq_{disp} X_1$, then for any distortion function $g$ we have $$\rho_g[S^{-}] = \rho_g[X_1] + \rho_{\overline{g}}[X_2]$$ 
where $\overline{g}$ is the dual distortion function of $g$. 
\end{thm}
\begin{proof}
First, we find that the aggregate loss function $\varphi^{-}(p)$ can be expressed as  
\begin{align*} 
\varphi^{-}(p) &= F^{-1}_{X_1}(p) + F^{-1}_{X_2}(1 - p) \\
&= F^{-1}_{X_1}(p) - F^{-1}_{-X_2}(p) \\ 
&= F^{-1}_{X_1}(p) - F^{-1}_{Y}(p) + 2F^{-1}_{X_2}(0.5), \,\, p \in (0, 1), 
\end{align*}
where $Y = 2F^{-1}_{X_2}(0.5) - X_2$. For the definition of $\varphi^{-}(p)$, see \cite{Chaoubi et al. (2020)}. Then, for any $y \in \mathbb{R}$, we compute 
\begin{align*} 
F_Y(y) &= \text{Pr}[2F^{-1}_{X_2}(0.5) - X_2 \leq y] \\
&= \text{Pr}[X_2 \geq 2F^{-1}_2(0.5) - y] \\
&= 1 - F_{X_2}(2F^{-1}_{X_2}(0.5) - y) \\
&= 1 - F_{X_2}(F^{-1}_{X_2}(0.5) + (F^{-1}_{X_2}(0.5)- y)) \\
&= F_{X_2}(F^{-1}_{X_2}(0.5) - (F^{-1}_{X_2}(0.5)- y)) = F_{X_2}(y), 
\end{align*} 
thanks to the symmetry and continuity of the distribution function $F_{X_2}$. This means that the variable $Y$ is equal to $X_2$ in the sense of distribution. Therefore, we have $$\varphi^{-}(p) = F^{-1}_{X_1}(p) - F^{-1}_{X_2}(p) + 2F^{-1}_{X_2}(0.5), \,\, p \in (0, 1).$$ By the condition $X_2 \leq_{disp} X_1$, we see that $\varphi^{-}(p)$ is a non-decreasing function on $(0, 1)$. Moreover, since the distribution functions $F_{X_i}(x), i = 1, 2$, are strictly increasing on $(F^{-1+}_{X_{i}}(0), F^{-1}_{X_i}(1))$, the inverse distribution functions $F^{-1}_{X_i}(p), i = 1, 2$ are continuous on $(0, 1)$, which in turn implies that $\varphi^{-}(p)$ is a continuous function on $(0, 1)$. So, by Theorem 1 in \cite{Dhaene et al. (2002a)}, we have 
\begin{align*}  
F^{-1}_{S^{-}}(p) = F^{-1}_{\varphi^{-}(U)}(p) = \varphi^{-}(F^{-1}_U(p)) = \varphi^{-}(p) = F^{-1}_{X_1}(p) + F^{-1}_{X_2}(1 - p)
\end{align*}
for any $p \in (0, 1)$. 

Second, for any distortion function $g$, by Theorem 7 in \cite{Dhaene et al. (2012)}, we know that there exist a left-continuous distortion function $g_l$, a right-continuous distortion function $g_r$, and two non-negative weights $c_l, c_r$ with $c_l + c_r= 1$ such that $g = c_lg_l + c_rg_r$ and 
\begin{align*} 
\rho_g[S^{-}] = c_l\rho_{g_l}[S^{-}] + c_r\rho_{g_r}[S^{-}]. 
\end{align*} 
For the left-continuous distortion function $g_l$, by Theorem 6 in \cite{Dhaene et al. (2012)}, the distortion risk measure $\rho_{g_l}[S^{-}]$ has the following Lebesgue-Stieltjes integral representation: $$\rho_{g_l}[S^{-}] = \int_{[0, 1]}F^{-1}_{S^{-}}(1 - q) dg_l(q).$$ 
From this, it follows that 
\begin{align*}
\rho_{g_l}[S^{-}] &= \int_{[0, 1]}F^{-1}_{S^{-}}(1 - q) dg_l(q) \\
&= \int_{[0, 1]}F^{-1}_{X_1}(1 - q) dg_l(q) + \int_{[0, 1]}F^{-1}_{X_2}(q) dg_l(q) \\
&= \int_{[0, 1]}F^{-1}_{X_1}(1 - q) dg_l(q) + \int_{[0, 1]}F^{-1}_{X_2}(1 - q) d\overline{g}_l(q) \\
&= \int_{[0, 1]}F^{-1}_{X_1}(1 - q) dg_l(q) + \int_{[0, 1]}F^{-1+}_{X_2}(1 - q) d\overline{g}_l(q) \\
&= \rho_{g_l}[X_1] + \rho_{\overline{g}_l}[X_2]
\end{align*}
where $\overline{g}_l$, the dual distortion function of $g_l$, is a right-continuous distortion function on $[0, 1]$. The equation $F^{-1+}_{X_2}(p) = F^{-1}_{X_2}(p), p \in (0, 1)$, is true because the distribution function $F_{X_2}(x)$ is strictly increasing on $(F^{-1+}_{X_{2}}(0), F^{-1}_{X_2}(1))$ while the integral representation $\rho_{\overline{g}_l}[X_2] = \int_{[0, 1]}F^{-1+}_{X_2}(1 - q) d\overline{g}_l(q)$ is due to Theorem 4 in \cite{Dhaene et al. (2012)}. Similarly, for the right-continuous distortion function $g_r$, by Theorem 4 in \cite{Dhaene et al. (2012)}, the distortion risk measure $\rho_{g_r}[S^{-}]$ has the following Lebesgue-Stieltjes integral representation: $$\rho_{g_r}[S^{-}] = \int_{[0, 1]}F^{-1+}_{S^{-}}(1 - q) dg_r(q).$$ Recall that, under the condition given in Theorem \ref{0265}, the aggregate loss function $\varphi^{-}(p) = F^{-1}_{X_1}(p) + F^{-1}_{X_2}(1- p)$ is a non-decreasing and continuous function on $(0, 1)$. Thus, for any $p \in (0, 1)$, by Theorem 1 in \cite{Dhaene et al. (2002a)} we have   
$$F^{-1+}_{S^{-}}(p) = F^{-1+}_{\varphi^{-}(U)}(p) = \varphi^{-}(F^{-1+}_U(p)) = \varphi^{-1}(p) = F^{-1}_{X_1}(p) + F^{-1}_{X_2}(1 - p). $$
It follows that 
\begin{align*}
\rho_{g_r}[S^{-}] &= \int_{[0, 1]}F^{-1}_{X_1}(1 - q) dg_r(q) + \int_{[0, 1]}F^{-1}_{X_2}(q) dg_r(q) \\
&= \int_{[0, 1]}F^{-1}_{X_1}(1 - q) dg_r(q) + \int_{[0, 1]}F^{-1}_{X_2}(1 - q) d\overline{g}_r(q) \\
&= \int_{[0, 1]}F^{-1+}_{X_1}(1 - q)] dg_r(q) + \int_{[0, 1]}F^{-1}_{X_2}(1 - q) d\overline{g}_r(q) \\
&= \rho_{g_r}[X_1] + \rho_{\overline{g}_r}[X_2], 
\end{align*}
where $\overline{g}_r$, the dual distortion function of $g_l$, is a left-continuous distortion function on $[0, 1]$. As above, the equation $F^{-1+}_{X_1}(p) = F^{-1}_{X_1}(p), p \in (0, 1)$, is true because the distribution function $F_{X_1}(x)$ is strictly increasing on $(F^{-1+}_{X_{1}}(0), F^{-1}_{X_1}(1))$. 

Finally, we compute  
\begin{align*}  
\rho_g[S^{-}] &= c_l\rho_{g_l}[S^{-}] + c_r\rho_{g_r}[S^{-}] \\ 
&= c_l(\rho_{g_l}[X_1] + \rho_{\overline{g}_l}[X_2]) + c_r(\rho_{g_r}[X_1] + \rho_{\overline{g}_r}[X_2]) \\  
&= (c_l\rho_{g_l}[X_1] + c_r\rho_{g_r}[X_1]) + (c_l\rho_{\overline{g}_l}[X_2] + c_r\rho_{\overline{g}_r}[X_2]) \\
&= \rho_g[X_1] + \rho_{\overline{g}}[X_2], 
\end{align*}
since $\overline{g} = c_l\overline{g}_l + c_r\overline{g}_r$ and $\rho_{\overline{g}}[X_2] = c_l\rho_{\overline{g}_l}[X_2] + c_r\rho_{\overline{g}_r}[X_2]$. This completes the proof. 
\end{proof} 
\begin{cor} \label{0155}
Assume that $(X^{-}_1, X^{-}_2)$ is a pair of normal random variables such that $X_i \sim N(\mu_i, \sigma^2_i), i = 1, 2$ and let $g$ be a distortion function.\begin{itemize} 
\item If $\sigma_1 \geq \sigma_2$, then $\rho_g[S^{-}] = \rho_g[X_1] + \rho_{\overline{g}}[X_2]$.  
\item If $\sigma_1 < \sigma_2$, then $\rho_g[S^{-}] = \rho_{\overline{g}}[X_1] + \rho_g[X_2]$.  
\end{itemize}
\end{cor}
\begin{cor} \label{0162}
Assume that $(X^{-}_1, X^{-}_2)$ is a pair of random variables such that $X_i \sim \text{Student}(\nu_i), i = 1, 2$ and let $g$ be a distortion function.
\begin{itemize} 
\item If $\nu_1 \leq \nu_2$, then $\rho_g[S^{-}] = \rho_g[X_1] + \rho_{\overline{g}}[X_2]$.  
\item If $\nu_2 < \nu_1$, then $\rho_g[S^{-}] = \rho_{\overline{g}}[X_1] + \rho_g[X_2]$.  
\end{itemize}
\end{cor}
In fact, by the well-known additivity property of distortion risk measures for comonotonic risks we are able to generalize Corollary \ref{0155} to the case that follows. 
\begin{cor} \label{0221}
Consider the sum $S = \sum^{m + n}_{i = 1}X_i$ where $X_i \sim N(\mu_i, \sigma_i^2), i = 1, 2, ..., m + n$. Assume that the random vector $(X_1, ..., X_{m + n})$ satisfies 
$$(X_1, ..., X_{m + n}) \stackrel{d}{=} (F_{X_1}^{-1}(U), ..., F_{X_m}^{-1}(U), F_{X_{m + 1}}^{-1}(1 - U), ..., F_{X_{m + n}}^{-1}(1 - U))$$ 
for a random variable $U \sim U(0, 1)$ and that $g$ is a distortion function. 
\begin{itemize} 
\item If $\sum^{m}_{i = 1}\sigma_i \geq \sum^{m + n}_{i = m + 1}\sigma_{i}$, then $\rho_g[S] = \sum^m_{i = 1}\rho_g[X_i] + \sum^{m + n}_{i = m + 1}\rho_{\overline{g}}[X_{i}]$. 
\item If $\sum^{m}_{i = 1}\sigma_i < \sum^{m + n}_{i = m + 1}\sigma_{i}$, then $\rho_g[S] = \sum^m_{i = 1}\rho_{\overline{g}}[X_i] + \sum^{m + n}_{i = m + 1}\rho_g[X_{i}]$. 
\end{itemize}
\end{cor}
\begin{proof}
Let $S_1 = \sum^m_{i = 1}F^{-1}_{X_i}(U)$ and $S_2 = \sum^{m + n}_{i = m + 1}F^{-1}_{X_{i}}(1 - U)$. Then it is easy to check that $S_1$ and $S_2$ are comonotonic sums and that $$S_1 \sim N\bigg(\sum^{m}_{i = 1}\mu_i, \Big(\sum^m_{i = 1}\sigma_i\Big)^2\bigg),\,\, S_2 \sim N\bigg(\sum^{m + n}_{i = m + 1}\mu_i, \Big(\sum^{m + n}_{i = m + 1}\sigma_i\Big)^2\bigg).$$ Moreover, we have $S \stackrel{d}{=} S_1 + S_2$. By the well-known additivity property of distortion risk measures for comonotonic risks, we have $$\rho_g[S_1] = \sum^m_{i = 1}\rho_g[X_i], \,\, \rho_{\overline{g}}[S_1] = \sum^m_{i = 1}\rho_{\overline{g}}[X_i]$$ 
and $$\rho_g[S_2] = \sum^{m + n}_{i = m + 1}\rho_g[X_i], \,\, \rho_{\overline{g}}[S_2] = \sum^{m + n}_{i = m + 1}\rho_{\overline{g}}[X_i].$$
Therefore, Corollary \ref{0221} follows form Theorem \ref{0265}. 
\end{proof}
\section{Applications}
It is well-known that distortion risk measures include as special cases a number of important risk measures, such as the VaR and TVaR risk measures. So, by using Theorem \ref{0265}, we can obtain a number of corollaries as follows. We remark that similar results are obtained in Proposition 1 and Proposition 2 in \cite{Chaoubi et al. (2020)} using a different method. 

\begin{cor} \label{0330}
Assume that $(X^{-}_1, X^{-}_2)$ is a pair of symmetric random variables such that their distribution functions $F_{X_i}(x), i = 1, 2$, are continuous and strictly increasing on $(F^{-1+}_{X_{i}}(0), F^{-1}_{X_i}(1))$.  If $X_2 \leq_{disp} X_1$, then  for any $p \in (0, 1)$ $$\text{VaR}_p[S^{-}] = \text{VaR}_p[X_1] + \text{VaR}_{1 - p}[X_2].$$
\end{cor}
\begin{proof}
For any fixed $p \in (0, 1)$, we consider the function $g$ defined by 
$$g(q) = \mathbb{I}(q > 1 - p), \,\, q \in [0, 1].$$ 
It is easy to check that $g$ is a left-continuous distortion function on $[0, 1]$. Therefore, by Theorem 6 in \cite{Dhaene et al. (2012)}, for any random variable $X$ the distortion risk measure $\rho_{g}[X]$ has the following Lebesgue-Stieltjes integral representation: $$\rho_{g}[X] = \int_{[0, 1]}F^{-1}_{X}(1 - q) dg(q).$$ 
Since $F^{-1}_X(q)$ is a left-continuous function on $(0, 1)$, it follows that 
$$ 
\rho_g[X] = \lim_{q \to (1 - p)^{+}}F^{-1}_{X}(1 - q)g(q) - \lim_{q \to (1 - p)^{-}}F^{-1}_{X}(1 - q)g(q) = \text{VaR}_p[X].
$$ Let $\overline{g}(x) = 1 - g(1 - x), x \in [0, 1]$, be the dual distortion function of $g$. Then it is easy to check that $\overline{g}$ is a right-continuous distortion function on $[0, 1]$. By Theorem 4 in \cite{Dhaene et al. (2012)}, for each $X_i, i = 1, 2$, the distortion risk measure $\rho_{\overline{g}}[X_i]$ has the following Lebesgue-Stieltjes integral representation: $$\rho_{\overline{g}}[X_i] = \int_{[0, 1]}F^{-1+}_{X_i}(1 - q) d\overline{g}(q).$$ Since $F_{X_i}(x)$ is strictly increasing on $(F^{-1+}_{X_{i}}(0), F^{-1}_{X_i}(1))$, we have $$F^{-1+}_{X_i}(1 - q) = F^{-1}_{X_i}(1 - q), \,\, q \in (0, 1),$$ and it follows that $$
\rho_{\overline{g}}[X_i] = \lim_{q \to p^{+}}F^{-1}_{X_i}(1 - q)\overline{g}(q) - \lim_{q \to p^{-}}F^{-1}_{X_i}(1 - q)\overline{g}(q) = \text{VaR}_{1 - p}[X_i], \,\, i = 1, 2.
$$
Now we can see that Corollary \ref{0330} immediately follows from Theorem \ref{0265}.  
\end{proof}

\begin{cor} \label{0338}
Assume that $(X^{-}_1, X^{-}_2)$ is a pair of symmetric random variables such that their distribution functions $F_{X_i}(x), i = 1, 2$, are continuous and strictly increasing on $(F^{-1+}_{X_{i}}(0), F^{-1}_{X_i}(1))$.  If $X_2 \leq_{disp} X_1$, then for any $p \in (0, 1)$ $$\text{TVaR}_p[S^{-}] = \text{TVaR}_p[X_1] + \text{LTVaR}_{1 - p}[X_2].$$  
\end{cor}
\begin{proof}
For any fixed $p \in (0, 1)$, we define the function $g$ by $$g(q) = \min\Big\{\frac{q}{1 - p}, 1\Big\}, \,\, q \in [0, 1].$$ It is easy to check that $g(q)$ is a continuous distortion function on $[0, 1]$. Therefore, $\overline{g}(q) = 1 - g(1 - q)$, the dual distortion function of $g$, is continuous on $[0, 1]$ as well. By Theorem 6 in \cite{Dhaene et al. (2012)}, for any random variable $X$ we have 
\begin{align*} 
\rho_g[X] &= \int_0^1\text{VaR}_{1 - q}[X]dg(q) \\
&= \frac{1}{1 - p}\int_0^{1 - p}\text{VaR}_{1 - q}[X]dq \\ 
&= \frac{1}{1 - p}\int_p^1\text{VaR}_q[X]dq \\ 
&= \text{TVaR}_{p}[X], \,\, p \in (0, 1), 
\end{align*} 
and 
\begin{align*} 
\rho_{\overline{g}}[X] &= \int_0^1\text{VaR}_{1 - q}[X]d\overline{g}(q) \\ 
&= \frac{1}{1 - p}\int_p^1 \text{VaR}_{1 - q}[X]dq \\ 
&= \frac{1}{1 - p}\int_0^{1 - p}\text{VaR}_q[X]dq \\ 
&= \text{LTVaR}_{1 - p}[X], \,\, p \in (0, 1). 
\end{align*}
So, Corollary \ref{0338} also follows from Theorem \ref{0265}. 
\end{proof}

\begin{cor} \label{0346}
Assume that $(X^{-}_1, X^{-}_2)$ is a pair of symmetric random variables such that their distribution functions $F_{X_i}(x), i = 1, 2$, are continuous and strictly increasing on $(F^{-1+}_{X_{i}}(0), F^{-1}_{X_i}(1))$.  If $X_2 \leq_{disp} X_1$, then for any $p \in (0, 1)$ $$\text{WT}_p[S^{-}] = \text{WT}_p[X_1] + \text{WT}_{1 - p}[X_2].$$
\end{cor}
\begin{proof}
For any fixed $p \in (0, 1)$, we define the distortion function $$g(q) = \Phi[\Phi^{-1}(q) + \Phi^{-1}(p)], \quad q \in [0, 1]$$ which is usually called the Wang Transform (distortion function) at level $p$ and the corresponding distortion risk measure $\rho_g[X]$ is called the Wang Transform risk measure, denoted by $\text{WT}_p[X]$, see Example 5.1.1 in \cite{Dhaene et al. (2006)}. It is easy to check that the function $g(q)$ is continuous on $[0, 1]$. This implies that the dual distortion function $\overline{g}(q)$ is also continuous on $[0, 1]$. Therefore, for any random variable $X$, by the Lebesgue-Stieltjes integral representation of $\rho_{\overline{g}}[X]$,  see Theorem 6 in \cite{Dhaene et al. (2012)}, we have  
\begin{align*} 
\rho_{\overline{g}}[X] &= \int_0^1 \text{VaR}_{1 - q}[X]d\overline{g}(q) \\ 
&= \int_0^1 \text{VaR}_{1 - q}[X]d(1 - \Phi[\Phi^{-1}(1 - q) + \Phi^{-1}(p)]) \\
&= \int_0^1 \text{VaR}_{1 - q}[X]d\Phi(\Phi^{-1}(q) - \Phi^{-1}(p)) \\ 
&= \int_0^1 \text{VaR}_{1 - q}[X]d\Phi(\Phi^{-1}(q) + \Phi^{-1}(1 - p)) \\
&= \text{WT}_{1 - p}[X]
\end{align*}
So, Corollary \ref{0346} immediately follows from Theorem \ref{0265}. 
\end{proof}
\begin{ex} 
Assume that $(X^{-}_1, X^{-}_2)$ is a pair of normal random variables such that $X_i \sim N(\mu_i, \sigma^2_i), i = 1, 2$ and let $p \in (0, 1)$ be a real number. If $\sigma_1 \geq \sigma_2$, then 
\begin{itemize} 
\item $\text{VaR}_p[S^{-}] = \text{VaR}_p[X_1] + \text{VaR}_{1 - p}[X_2]$; 
\item $\text{TVaR}_p[S^{-}] = \text{TVaR}_p[X_1] + \text{LTVaR}_{1 - p}[X_2]$; 
\item $\text{WT}_p[S^{-}] = \text{WT}_p[X_1] + \text{WT}_{1 - p}[X_2]$. 
\end{itemize}
If $\sigma_1 < \sigma_2$, then  
\begin{itemize} 
\item $\text{VaR}_p[S^{-}] = \text{VaR}_{1 - p}[X_1] + \text{VaR}_{p}[X_2]$; 
\item $\text{TVaR}_p[S^{-}] = \text{LTVaR}_{1 - p}[X_1] + \text{TVaR}_{p}[X_2]$; 
\item $\text{WT}_p[S^{-}] = \text{WT}_{1 - p}[X_1] + \text{WT}_{p}[X_2]$. 
\end{itemize}

\end{ex}
\begin{ex} \label{0354}
Assume that $(X^{-}_1, X^{-}_2)$ is a pair of random variables such that $X_i \sim \text{Student}(\nu_i), i = 1, 2$. Let $p \in (0, 1)$ be a real number. 
If $\nu_1 \leq \nu_2$, then 
\begin{itemize} 
\item $\text{VaR}_p[S^{-}] = \text{VaR}_p[X_1] + \text{VaR}_{1 - p}[X_2]$; 
\item $\text{TVaR}_p[S^{-}] = \text{TVaR}_p[X_1] + \text{LTVaR}_{1 - p}[X_2]$; 
\item $\text{WT}_p[S^{-}] = \text{WT}_p[X_1] + \text{WT}_{1 - p}[X_2]$. 
\end{itemize}
If $\nu_1 > \nu_2$, then  
\begin{itemize} 
\item $\text{VaR}_p[S^{-}] = \text{VaR}_{1 - p}[X_1] + \text{VaR}_{p}[X_2]$; 
\item $\text{TVaR}_p[S^{-}] = \text{LTVaR}_{1 - p}[X_1] + \text{TVaR}_{p}[X_2]$; 
\item $\text{WT}_p[S^{-}] = \text{WT}_{1 - p}[X_1] + \text{WT}_{p}[X_2]$. 
\end{itemize}
\end{ex}

\section{Further discussions}
In this section, we want to study distortion risk measures, including VaR's and TVaR's in particular, of the sum of two counter-monotonic variables in the case where the marginals under consideration are assumed to follow log-normal distributions. We will show that, even in this concrete case, distortion risk measures of the counter-monotonic sum $S^{-}$ are generally not easy to calculate. The main reason is that the aggregate loss function $\varphi^{-}(u)$ is no longer monotonic on the whole interval $(0, 1)$ but first strictly decreases and then strictly increases displaying a U-shape behavior. 
\begin{thm} \label{0272}
Assume that $(X^{-}_1, X^{-}_2)$ is a pair of identically distributed log-normal random variables. Then, for any left-continuous distortion function $g$, we have $$\rho_{g}[S^{-}] = \int^1_0 (F^{-1}_{X_1}({q}/{2}) + F^{-1}_{X_1}(1 - {q}/{2}))dg(q).$$
\end{thm} 
\begin{proof}
For any left-continuous distortion function $g$, by Theorem 6 in \cite{Dhaene et al. (2012)}, we know that the distortion risk measure $\rho_{g}[S^{-}]$ has the following Lebesgue-Stieltjes integral representation: $$\rho_{g}[S^{-}] = \int^1_{0}F^{-1}_{S^{-}}(1 - q) dg(q).$$ So, it suffices to show that $$F^{-1}_{S^{-}}(1 - q) = F^{-1}_{X_1}(q/2) + F^{-1}_{X_1}(1 - q/2)$$
for any $q \in (0, 1)$, which is, however, fairly easy to check if we are concerned with Proposition 2 in \cite{Chaoubi et al. (2020)}. Indeed, the aggregate loss function becomes $$\varphi^{-}(u) = F^{-1}_{X_1}(u) + F^{-1}_{X_1}(1 - u), \,\, u \in(0, 1)$$ 
which is, obviously, symmetric around $u_0 = 1/2$. Moreover, easy computations show that $\varphi^{-}(u)$ is strictly decreasing on $(0, 1/2)$ and strictly increasing on $(1/2, 1)$, so that it reaches its minimum value at the point $u_0 = 1/2$. Now, Theorem \ref{0272} immediately follows from Proposition 2 in \cite{Chaoubi et al. (2020)}. We remark that for Proposition 2 in \cite{Chaoubi et al. (2020)} to hold the function $\varphi^{-}(u)$ does not need to be convex and the monotonicity will be sufficient. 
\end{proof}

\begin{cor} \label{0280}
Assume that $(X^{-}_1, X^{-}_2)$ is a pair of identically distributed log-normal random variables such that $X_1 \stackrel{d}{=} X_2 \sim LN(\mu, \sigma^2)$. Then, for any $p \in (0, 1)$, we have 
\begin{equation} \label{0282}
\text{VaR}_p[S^{-}] = e^{\mu + \sigma\Phi^{-1}(\frac{1 - p}{2})} + e^{\mu + \sigma\Phi^{-1}(\frac{1 + p}{2})}
\end{equation} and 
\begin{equation} \label{0285}
\text{TVaR}_p[S^{-}] = \frac{2e^{\mu + \frac{1}{2}\sigma^2}}{1 - p}\Big[\Phi\big(\Phi^{-1}(\frac{1 - p}{2}) - \sigma\big) + \Phi\big(\Phi^{-1}(\frac{1 - p}{2}) + \sigma\big)\Big].
\end{equation}
\end{cor}
\begin{proof} 
Corollary \ref{0280} follows from Theorem \ref{0272}. In fact, 
for any fixed $p \in (0, 1)$, we consider the function $g$ defined by $$g(q) = \mathbb{I}(q > 1 - p), \,\, q \in [0, 1].$$ Then we have $ \rho_g[S^{-}] = \text{VaR}_p[S^{-}].$ Moreover, by Theorem \ref{0272} we find that 
\begin{align*} 
\rho_g[S^{-}] &= \int^1_0(F^{-1}_{X_1}(\frac{q}{2}) + F^{-1}_{X_1}(1 - \frac{q}{2}))dg(q) \\ 
&= F^{-1}_{X_1}(\frac{1 - p}{2}) + F^{-1}_{X_1}(\frac{1 + p}{2}) \\ 
&= e^{\mu + \sigma\Phi^{-1}(\frac{1 - p}{2})} + e^{\mu + \sigma\Phi^{-1}(\frac{1 + p}{2})}
\end{align*}
from which equation $($\ref{0282}$)$ follows. To show the identity regarding TVaR's, for any fixed $p \in (0, 1)$, we consider the function $g$ defined by $$g(q) = \min\Big\{\frac{q}{1 - p}, 1\Big\}, \,\, q \in [0, 1].$$ Then we have $\rho_g[S^{-}] = \text{TVaR}_{p}[S^{-}].$ Moreover, by Theorem \ref{0272} we find that 
\begin{align*} 
\rho_g[S^{-}] &= \int^1_0(F^{-1}_{X_1}(\frac{q}{2}) + F^{-1}_{X_1}(1 - \frac{q}{2}))dg(q) \\ 
&= \frac{1}{1 - p}\int^{1 - p}_0(F^{-1}_{X_1}(\frac{q}{2}) + F^{-1}_{X_1}(1 - \frac{q}{2}))dq \\ 
&= \frac{e^{\mu}}{1 - p}\int^{1 - p}_0 [e^{\sigma\Phi^{-1}(\frac{q}{2})} + e^{-\sigma\Phi^{-1}(\frac{q}{2})}]dq \\ 
&= \frac{2e^{\mu}}{1 - p} \int^{\Phi^{-1}(\frac{1 - p}{2})}_{-\infty}(e^{\sigma t} + e^{-\sigma t})\Phi^{\, '}(t)dt \\ 
&= \frac{2e^{\mu + \frac{1}{2}\sigma^2}}{1 - p}\int^{\Phi^{-1}(\frac{1 - p}{2})}_{-\infty}\frac{1}{\sqrt{2\pi}}\big(e^{-\frac{1}{2}(t - \sigma)^2} + e^{-\frac{1}{2}(t + \sigma)^2}\big)dt \\ 
&= \frac{2e^{\mu + \frac{1}{2}\sigma^2}}{1 - p}\bigg(\int^{\Phi^{-1}(\frac{1 - p}{2}) - \sigma}_{-\infty}\frac{1}{\sqrt{2\pi}}e^{-\frac{1}{2}x^2}dx + \int^{\Phi^{-1}(\frac{1 - p}{2}) + \sigma}_{-\infty}\frac{1}{\sqrt{2\pi}}e^{-\frac{1}{2}y^2}dy\bigg) \\
&= \frac{2e^{\mu + \frac{1}{2}\sigma^2}}{1 - p}\Big[\Phi\big(\Phi^{-1}(\frac{1 - p}{2}) - \sigma\big) + \Phi\big(\Phi^{-1}(\frac{1 - p}{2}) + \sigma\big)\Big]
\end{align*}
from which equation $($\ref{0285}$)$ follows. 
\end{proof} 
We remark that the results in Corollary \ref{0280} are first obtained in Example 9 in \cite{Chaoubi et al. (2020)} using a different method. Now we consider the case where the log-normal marginals are not necessarily identically distributed. 

\begin{lem} \label{0323}
Assume that $(X^{-}_1, X^{-}_2)$ is a pair of non-identically distributed log-normal random variables such that $X_i \sim LN(\mu_i, \sigma_i^2), i = 1, 2$. Then, the aggregate loss function $\varphi^{-}$ is first strictly decreasing and then strictly increasing and it reaches its minimal value at the point $\Phi(\frac{\mu_2 - \mu_1 + \ln(\frac{\sigma_2}{\sigma_1})}{\sigma_1 + \sigma_2}) \in (0, 1)$. 
\end{lem}
For a proof of this result, we refer the interested reader to Proposition 10 in \cite{Chaoubi et al. (2020)}. By Lemma \ref{0323} we know that, for any $x \in (x^{\min}, x^{\max})$ with $x^{\min} = \inf_{u \in [0, 1]}\varphi^{-}(u)$ and $x^{\max} = \sup_{u \in [0, 1]}\varphi^{-}(u)$, there are always two elements satisfying the condition of Definition 3.2 in \cite{HLD}. Thus, we have $N_x = 2$ and $E_x = \{u_{x, 1}, u_{x, 2}\}$ with $u_{x, 1} < u_{x, 2}$ for all $x \in (x^{\min}, x^{\max})$. For the details about the notations $N_x$ and $E_x$ see Definition 3.2 in \cite{HLD}. 
\begin{cor} \label{0317} 
Assume that $(X^{-}_1, X^{-}_2)$ is a pair of non-identically distributed log-normal random variables. For any $p \in (0, 1)$,  if we let $u_{p, 1} < u_{p, 2}$ be the two elements of $E_{F^{-1}_{S^{-}}(p)}$, then we have $$\text{VaR}_{p}[S^{-}] = \text{VaR}_{u_{p, j}} + \text{VaR}_{1 - u_{p, j}}[X_2], \,\, j = 1, 2.$$
\end{cor} 

\begin{cor} \label{0320} 
Assume that $(X^{-}_1, X^{-}_2)$ is a pair of non-identically distributed log-normal random variables. For any $p \in (0, 1)$ and $\alpha \in [0, 1]$, if we let $u^{\alpha}_{p, 1} < u^{\alpha}_{p, 2}$ be the two elements of $E_{F^{-1(\alpha)}_{S^{-}}(p)}$, then we have 
\begin{align*} 
\text{TVaR}_{p}[S^{-}] &= \frac{1}{1 - p}u^{\alpha}_{p, 1}(LTVaR_{u^{\alpha}_{p, 1}}[X_1] + TVaR_{1 - u^{\alpha}_{p, 1}}[X_2]) \\ 
& \quad + \frac{1}{1 - p}(1 - u^{\alpha}_{p, 2})(TVaR_{u^{\alpha}_{p, 2}}[X_1] + LTVaR_{1 - u^{\alpha}_{p, 2}}[X_2]) \\ 
&\quad + \frac{1}{1 - p}F^{-1(\alpha)}_{S^{-}}(u^{\alpha}_{p, 2} - u^{\alpha}_{p, 1} - p).
\end{align*}
\end{cor}
For the proof of Corollary \ref{0317} and Corollary \ref{0320} the reader is referred to the equations (4.1) and (5.1) in \cite{HLD}, where the authors successfully treated the problem in a very general case. 

%%%%%%%%%%%%%%%%%%%%%%%%%%%

\end{document}